\documentclass[11pt, letterpaper]{article}
\usepackage{amsmath, amssymb, amsthm, amsfonts, bbm, commath}
\usepackage[inline]{enumitem}
\usepackage{tikz}
\usepackage{subcaption}
\usepackage{authblk}
\usepackage{graphicx}
\usepackage{hyperref}
\usepackage{xcolor}
\usepackage{cancel}

\title{
	Exact Byzantine Consensus Under Local-Broadcast Model
	\thanks{This research is supported in part by National Science Foundation award 1409416, and Toyota InfoTechnology
		Center. Any opinions, findings, and conclusions or recommendations expressed here are those of the authors and do not necessarily reflect the views of the funding agencies or the U.S. government.
	}
}

\author[1]{
	Syed Shalan Naqvi
}

\author[1]{
	Muhammad Samir Khan
}

\author[2]{
	Nitin H. Vaidya
}

\affil[1]{
	Department of Computer Science \protect\linebreak
	University of Illinois at Urbana-Champaign \protect\linebreak
	\texttt{\{naqvi5,mskhan6\}@illinois.edu} \protect\linebreak
}

\affil[2]{
	Department of Computer Science \protect\linebreak
	Georgetown University \protect\linebreak
	\texttt{nv198@georgetown.edu}
}

\newtheorem{theorem}{Theorem}[section]
\newtheorem{lemma}[theorem]{Lemma}

\newtheorem{definition}[theorem]{Definition}

\renewenvironment{proof}{\noindent{\bf Proof:} \hspace*{1mm}}{
	\hspace*{\fill} $\Box$ }

\newcommand{ \connect }[1]{
	\stackrel{#1}{ \implies }
}
\newcommand{ \notconnect }[1]{
	\cancel{ \stackrel{#1}{ \implies } }
}
\newcommand{\floor}[1]{
	\lfloor #1 \rfloor
}
\newcommand{\ceil}[1]{
	\lceil #1 \rceil
}

\begin{document}
	\maketitle

	\section{Introduction}
This paper considers the problem of achieving exact Byzantine consensus in a synchronous system under a {\em local-broadcast} communication model.
The nodes communicate with each other via message-passing. The communication network is modeled as an undirected graph, with each vertex representing a
node in the system. Under the {\em local-broadcast} communication model, when any node transmits a message, all its neighbors in the communication graph receive the message reliably. This
communication model is motivated by wireless networks. For instance, in the communication graph in Figure 1, node 1 has neighbors 2 and 3, whereas
node 2 has neighbors 1, 4 and 5. Thus, when node 1 transmits a message, the message is received by both its neighbors, namely 2 and 3.
Similarly, when node 2 transmits a message, it is received by its neighbors 1, 4 and 5.

In this brief announcement,
we present necessary and sufficient conditions on the underlying communication graph to achieve exact Byzantine consensus under the local-broadcast communication model.

	\section{Related Work}
Under the classical point-to-point network model, $n\geq 3f+1$ and vertex connectivity $\geq 2f+1$ are both necessary and sufficient conditions for Byzantine consensus on
undirected graphs \cite{dolev1982byzantine,fischer1986easy}.
Closest to our contributions is the work on consensus using {\em partial broadcast channels}
that help achieve various forms of ``non-equivocation''
\cite{ravikant2004byzantine,jaffe2012price,li2016towards}.
Intuitively, non-equivocation prevents a node from transmitting inconsistent messages.
Prior work has considered various constraints on equivocation, however, surprisingly, the local-broadcast model 
considered in this paper does not seem to have been addressed yet. In a completely connected graph, a broadcast channel ensures that a node cannot send inconsistent messages to
any pair of nodes, and it is easy to see that $2f+1$ nodes suffice to tolerate $f$ Byzantine faults. \cite{ravikant2004byzantine} consider
networks modeled as a $(2,3)$-uniform hypergraph. In this setting, each transmission is viewed as occurring on one of the hyperedges, and the nodes belonging to the hyperedge
receive the message reliably. Motivated by the work in \cite{ravikant2004byzantine},
\cite{jaffe2012price} obtained further results on networks with 3-uniform hypergraphs.\cite{li2016towards} studies the problem of iterative approximate
Byzantine consensus using 3-hyperedges in which one of nodes in each hyperedge is identified as the unique sender.

For the local-broadcast communication model, \cite{koo2006reliable,bhandari2005podc,pelc2005broadcasting} consider
the problem of reliable {\em broadcast}. Although their communication model is analogous to this paper, we consider Byzantine {\em consensus}, whereas this prior work considers
reliable {\em broadcast}. It turns out that the network requirements for these two problems are quite different.

Another related line of research considers approximate Byzantine consensus algorithms
with a restrictive computation structure. In particular, each node
maintains a real-valued state, and in each iteration of the algorithm, a node
computes its new state as a {\em weighted linear combination} of its previous state and its neighbors' state. For this class of iterative approximate consensus algorithms, necessary and sufficient conditions
under the local-broadcast model have been identified \cite{zhang2012robustness}. These conditions are distinct from those obtained in our work, since we do not constrain the algorithm structure; additionally, we consider exact consensus.

	\section{System Model}
$G=(V,E)$ represents the communication graph. Each vertex $u\in V$ represents node $u$, and edge $(u,v)\in E$ if and only if nodes $u$ and $v$ can receive each other's message
transmissions. The system is synchronous.
The local-broadcast communication model is assumed. Thus, a transmission by any node $u\in V$ is received reliably (and identically) by all the nodes in the set
$\{v~|~(u,v)\in E\}$.
It is assumed that when a node $u$ receives a message, it can correctly identify the neighbor that sent the message.

Each node has a binary input, i.e., the input is in $\{0,1\}$. Up to $f$ of the nodes may be Byzantine faulty. A faulty node is assumed to have complete knowledge of the states of
all the nodes, the algorithm, and the communication graph.
A correct algorithm for Byzantine consensus must satisfy the following conditions:
\begin{itemize}
	\item \textbf{Agreement:} All non-faulty nodes decide on an identical value.
	\item \textbf{Validity:} If input of all the non-faulty nodes equals $b\in\{0,1\}$, then non-faulty nodes decide $b$.
	\item \textbf{Termination:} The algorithm terminates after a bounded duration of time.
\end{itemize}

\includegraphics[scale=0.5]{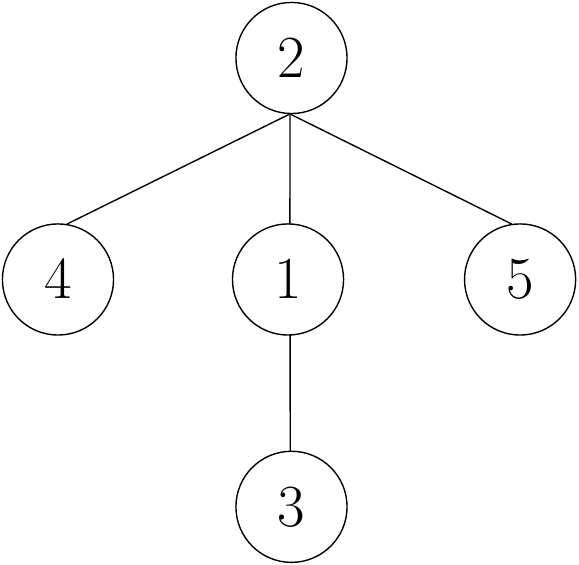}
\hspace*{1.2in}
\includegraphics[scale=0.6]{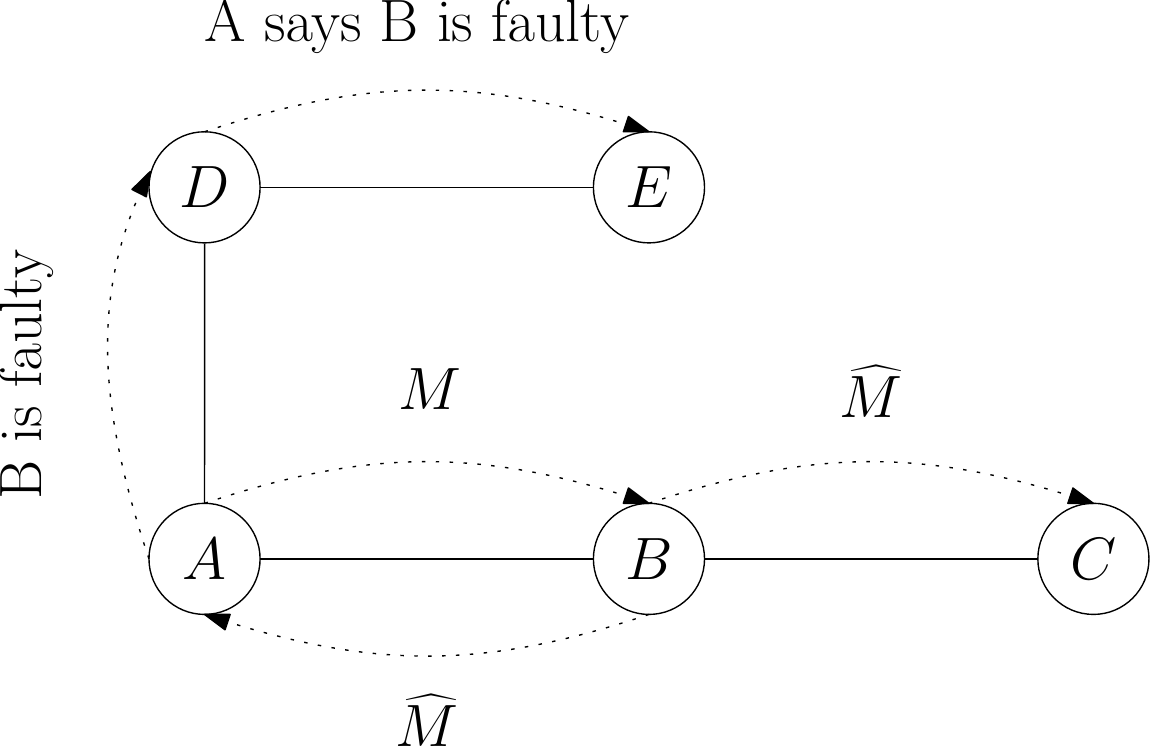}

Figure 1: Example graph \hspace*{1.2in}
Figure 2: Detecting message tampering

	\section{Our Results}
Before presenting the results, we illustrate a benefit of the local-broadcast model.
As an example, consider the illustration in Figure 2, where we assume that node B is Byzantine faulty, and other
nodes are non-faulty. In the figure, solid lines show edges between nodes, and directed dotted lines show the messages being transmitted. As shown in Figure 2, node A
transmits message $M$ to node B, and then expects B to forward $M$ to node C. However, node B tampers the message before forwarding to C. Instead of $M$, node B forwards message
$\widehat{M}$. Due to the local-broadcast property, node A will also receive $\widehat{M}$, and learn that node B is faulty. Subsequently, node A sends a message to node E along
path A-D-E claiming that node B is faulty\footnote{Some messages have been omitted from the figure for brevity.}. If the path travelled by the message is also included in the message, then node
E can now infer that at least one of the nodes on the path D-A-B must be faulty. Such inferences
under the local-broadcast model can be shown
to reduce the requirements on the communication graph, compared to the classical point-to-point communication model.
In particular, we have obtained the following results \cite{naqvi18thesis}.

\begin{theorem}[Necessary Conditions]
	\label{thm:necessary}
	Under the local-broadcast model, Byzantine consensus is impossible
	with up to $f$ Byzantine faults ($f>0$) if either condition below holds
	true:
	\begin{itemize}
		\item Minimum degree of $G(V,E)$ is less than $2f$, or
		\item Vertex connectivity of $G(V,E)$ is at most $\lfloor{3f/2}\rfloor$.
	\end{itemize}
\end{theorem}

\vspace*{4pt}

\begin{theorem}[Sufficient Conditions]
	\label{thm:sufficient}
	Under the local-broadcast model, Byzantine consensus can be achieved
	with up to $f$ Byzantine faults ($f>0$) if the vertex connectivity of
	$G(V,E)$ is at least $2f$.
	
\end{theorem}

The sufficient condition above is proved constructively by providing a correct Byzantine consensus algorithm
\cite{naqvi18thesis}.
A simpler algorithm for $2f$-connectivity is given in Appendix \ref{section alternate algorithm 2f}.
There is a gap between the necessary and sufficient conditions in Theorems \ref{thm:necessary} and \ref{thm:sufficient}, due the different thresholds on vertex connectivity. The following result suggests that it may be possible to significantly tighten the sufficient condition in Theorem \ref{thm:sufficient}.

\begin{theorem}
	\label{thm:graph}
	There exists a graph with vertex connectivity $\lfloor{\frac{3f}{2}}\rfloor+1$ on which Byzantine consensus
	with up to $f$ Byzantine faults is achievable under the local-broadcast model. 
\end{theorem}

Motivated by the above result, we are presently working to bridge the gap between Theorems
\ref{thm:necessary} and \ref{thm:sufficient}.

	\bibliographystyle{abbrv}
	\bibliography{ref}

\begin{thebibliography}{10}

\bibitem{bhandari2005podc}
V.~Bhandari and N.~H. Vaidya.
\newblock On reliable broadcast in a radio network.
\newblock In {\em ACM Symposium on Principles of Distributed Computing (PODC)},
  2005.

\bibitem{dolev1982byzantine}
D.~Dolev.
\newblock The {Byzantine} generals strike again.
\newblock {\em Journal of algorithms}, 3(1):14--30, 1982.

\bibitem{fischer1986easy}
M.~J. Fischer, N.~A. Lynch, and M.~Merritt.
\newblock Easy impossibility proofs for distributed consensus problems.
\newblock {\em Distributed Computing}, 1(1):26--39, 1986.

\bibitem{jaffe2012price}
A.~Jaffe, T.~Moscibroda, and S.~Sen.
\newblock On the price of equivocation in byzantine agreement.
\newblock In {\em Proceedings of the 2012 ACM symposium on Principles of
  distributed computing}, pages 309--318. ACM, 2012.

\bibitem{koo2006reliable}
C.-Y. Koo, V.~Bhandari, J.~Katz, and N.~H. Vaidya.
\newblock Reliable broadcast in radio networks: The bounded collision case.
\newblock In {\em Proceedings of the twenty-fifth annual ACM symposium on
  Principles of distributed computing}, pages 258--264. ACM, 2006.

\bibitem{li2016towards}
C.~Li, M.~Hurfin, Y.~Wang, and L.~Yu.
\newblock Towards a restrained use of non-equivocation for achieving iterative
  approximate {Byzantine} consensus.
\newblock In {\em Parallel and Distributed Processing Symposium, 2016 IEEE
  International}, pages 710--719. IEEE, 2016.

\bibitem{naqvi18thesis}
S.~S. Naqvi.
\newblock Exact {Byzantine} consensus under local-broadcast channels,
  University of Illinois at Urbana-Champaign, M.S. Thesis (Advisor: Nitin
  Vaidya), 2018.

\bibitem{pelc2005broadcasting}
A.~Pelc and D.~Peleg.
\newblock Broadcasting with locally bounded byzantine faults.
\newblock {\em Information Processing Letters}, 93(3):109--115, 2005.

\bibitem{ravikant2004byzantine}
D.~Ravikant, V.~Muthuramakrishnan, V.~Srikanth, K.~Srinathan, and C.~P. Rangan.
\newblock On {Byzantine} agreement over (2, 3)-uniform hypergraphs.
\newblock In {\em International Symposium on Distributed Computing}, pages
  450--464. Springer, 2004.

\bibitem{zhang2012robustness}
H.~Zhang and S.~Sundaram.
\newblock Robustness of information diffusion algorithms to locally bounded
  adversaries.
\newblock In {\em American Control Conference (ACC)}, 2012.

\end{thebibliography}

	\appendix


	\section{Algorithm for $2f$-Connectivity ($f > 0$)}	\label{section alternate algorithm 2f}
\begin{definition} \label{definition reliable receive}
	For two nodes $u, v \in V$, $v$ reliably receives a message sent by $u$ if
	\begin{enumerate}[topsep=0pt, itemsep=0pt]
		\item $u = v$,
		\item $v$ is a neighbor of $u$, or
		\item $v$ receives the message identically on at least $f + 1$ node disjoint $uv$-paths.
	\end{enumerate}
\end{definition}

Observe that if a node $u$ sends a message $m$, then any node $v$ can not reliably receive a message from $u$ other than $m$.

For simplicity, fix $2f$ node disjoint paths for each pair of nodes.
We say that a non-faulty node is a \emph{type A} node if it knows the identity of all $f$ faulty nodes.
Every other non-faulty node is a \emph{type B} node.
Initially, all non-faulty nodes are type B nodes.
As the algorithm is run, some non-faulty nodes will discover the identity of faulty nodes and transition to type A nodes.
We show that by the end of the algorithm, all type A nodes have reached consensus on the output.
Since type B nodes know the identity of all faulty nodes, they can ignore messages from paths with faulty nodes (faulty paths) and receive the correct decision value of type A nodes.

The algorithm proceeds in 3 rounds.
In round 1, each node $u \in V$ sends its input value to the entire network (i.e. ``floods'' its input value).
As all the communication is synchronous, $u$ can wait for $n$ synchronous time steps to ensure that the input value is propagated to each node in the graph.
Since the communication is via broadcast, all non-faulty neighbors of $u$ reliably receive the same value from $u$ in round 1, even if $u$ is faulty.
Therefore, we assume that the input value of a faulty node is the value it floods in the first round.
In round 2, for each node $v \in V$, its neighbors report on the messages $v$ propagated in round 1.
As before, after $n$ synchronous time steps, this information has been propagated to the entire network.
Again, due to the broadcast model, each neighbor of $v$ knows the exact messages propagated by $v$ in round $1$.
These messages help identify faulty nodes if they exhibit faulty behaviour (i.e. tamper with messages).
The details are in the proof of Lemma \ref{lemma faulty node can not hide}.
In round 3, type B nodes reach consensus and inform the type A nodes on the decision.
If no type B nodes exist, then type A nodes decide by themselves.

The algorithm is as follows.

\begin{enumerate}[label=Round \arabic*:,topsep=0pt, itemsep=0pt]
	\item (Flood)
	Each node $u$ floods its input value.

	\item (Report)
	For each node $v \in V$, its neighbors report on the messages propagated by $v$ in round 1.
	At the end, each node $w \in V$, attempts to discover the faulty nodes (details in proof of Lemma \ref{lemma faulty node can not hide}).

	\item (Decide)
	Fix an arbitrary non-faulty node $u$.
	If $u$ is a type B node, then it decides from the input values it received reliably in round 1, by taking the majority, and then floods the decision value.
	If $u$ is a type A node, then it waits for a decision value from a non-faulty (type B) node.
	If a decision value is received from a non-faulty (type B) node, then $u$ decides on this value.
	Otherwise, $u$ decides from the input values of all non-faulty nodes by taking the majority.
\end{enumerate}

\begin{lemma} \label{lemma faulty node can not hide}
	Message sent by a faulty node is received reliably by every node.
\end{lemma}

\begin{proof}
	Let $u$ be an arbitrary faulty node and let $v$ be an arbitrary node in the graph.
	If $u = v$ or $v$ is a neighbor of $u$, then the claim is trivially true from Definition \ref{definition reliable receive}.
	Otherwise, there exist $2f$ node disjoint paths from $u$ to $v$ as the graph is $2f$-connected.
	Since $u$ is faulty, only $f-1$ of these paths can have faulty nodes.
	Moreover, due to the broadcast model, $u$ sends the same message on all the $2f$ node disjoint paths.
	Therefore $v$ receives identical messages from the remaining $f+1$ node disjoint paths that do not have a faulty node.
\end{proof}

\begin{lemma} \label{lemma identify faulty nodes}
	Let $v$ and $w$ be distinct nodes.
	In round 1, if $v$ reliably receives an input value of some node $u \in V$ and $w$ does not, then $v$ knows the identity of $f$ faulty nodes after round $2$.
\end{lemma}

\begin{proof}
	Fix a node $u$ with input $b \in \set{ 0, 1 }$.
	Observe that, by Lemma \ref{lemma faulty node can not hide}, $u$ is a non-faulty node since $w$ did not receive $b$ from $u$ reliably in round 1.
	Note also that $v$ reliably receives $b$ from $u$ in round 1.
	Let $P_1, \dots, P_{2f}$ be the $2f$ node disjoint $uw$-paths.
	Since $w$ did not reliably receive $b$ from $u$ in round 1, therefore exactly $f$ of these paths have faulty nodes.
	WLOG let these paths be $P_1, \dots, P_f$.
	By Lemma \ref{lemma faulty node can not hide}, after round 2, $v$ reliably receives that some nodes on $P_1 \dots, P_f$ forwarded $\bar{b}$ in round 1.
	For each path, $v$ sets the first such node to be faulty.

	To see why this assignment is correct, consider an arbitrary path in $P_1, \dots, P_f$.
	WLOG let this path be $P_1$.
	Let $z$ be the faulty node in $P_1$ that tampers the message.
	Observe that each of $P_1, \dots, P_f$ has exactly one faulty node that tampers the message.
	In round 2, by Lemma \ref{lemma faulty node can not hide}, $v$ reliably receives that $z$ forwarded $\bar{b}$ in round 1.
	Moreover, let $x$ be an arbitrary node in $P_1$ before $z$.
	Then $x$ is non-faulty and forwarded $b$ in round 1.
	Therefore, in round 2, $v$ can not reliably receive that $x$ forwarded $\bar{b}$ in round 1.
\end{proof}

\begin{lemma} [Agreement] \label{lemma agreement}
	All type B nodes reliably receive the same set of input values in round 1.
\end{lemma}

\begin{proof}
	Suppose, for the sake of contradiction, that two nodes $v$ and $w$ are type B nodes and there exists a node $u$ such that $v$ reliably receives the input of $u$ and $w$ does not.
	Then, by Lemma \ref{lemma identify faulty nodes}, $v$ knows the identity of $f$ faulty nodes after round 2.
	Therefore, $v$ must be a type A node, a contradiction.
\end{proof}

\begin{lemma} [Validity] \label{lemma validity}
	Each node reliably receives input values of at least $2f$ other nodes.
\end{lemma}

\begin{proof}
	Since the graph is $2f$ connected, therefore each node has at least $2f$ neighbors.
	By Definition \ref{definition reliable receive} each node reliably receives input from these $2f$ nodes.
\end{proof}

\begin{theorem}
	The algorithm given above achieves Byzantine consensus with up to $f$ Byzantine faults ($f > 0$) if the vertex connectivity of $G(V, E)$ is at least $2f$.
\end{theorem}

\begin{proof}
	The termination follows from the construction of the algorithm.
	For validity and agreement, there are two cases to consider.

	\textbf{There is at least one type B node:}
	For agreement, note that by Lemma \ref{lemma agreement} all type B nodes receive the same input values and therefore decide on the same value by taking the majority.
	Type A nodes know the identity of all $f$ ($f > 0$) faulty nodes and so they can ignore messages on paths with faulty nodes.
	So we only need one non-faulty path between a type A node and a type B node.
	Since there are at least $f$ node disjoint paths without any faulty nodes between any two nodes, therefore in round 3 each type A node correctly receives a message from a type B node about the final decision and decides on the same value.
	For validity, note that by Lemma \ref{lemma validity} each type B node is aware of input values of at least $2f + 1$ nodes (including its own).
	If the input of all non-faulty nodes equals $b \in \set{0, 1}$, then by taking the majority a type B node will decide on $b$.

	\textbf{There are no type B nodes:}
	Let $u$ be an arbitrary non-faulty node.
	Since there are no type B nodes, $u$ is a type A node that does not receive any decision value from any type B node in round 3.
	Since $u$ knows the identity of $f$ faulty nodes, $u$ can check messages received from non-faulty nodes in round 1 and ignore messages from any path that contains a faulty node.
	Therefore, $u$ knows the input values of all non-faulty nodes.
	For agreement, observe that all non-faulty nodes are type A and each non-faulty node decides on the same decision value by taking the majority of the input values of non-faulty nodes.
	For validity, observe that all non-faulty nodes only consider the input values of non-faulty nodes.
\end{proof}

	\section{Alternate characterization of the Necessity Condition}
Consider an undirected graph $G = (V, E)$.
For a set $S \subset V$, define $\Gamma(S) := \set{ v \in V - S \mid \exists u \in S, uv \in E }$.
For two disjoint sets $S, T \subset V$, define $T \connect{k} S$ if $\envert[0]{ \Gamma(S) \cap T } \ge k$.
For a set $F \subset V$, $\del{ L, C, R }$ is an \emph{$F$-partition} of $G$ if
\begin{enumerate}[topsep=0pt, itemsep=0pt]
	\item $\del{ L, C, R }$ is a partition of $V$,
	\item $L - F$ is non-empty, and
	\item $R - F$ is non-empty.
\end{enumerate}
$G$ is \emph{$f$-good} if for every set $F \subset V$ of cardinality at most $f$, every $F$-partition $\del{ L, C, R }$ of $V$ is such that
\begin{enumerate}[topsep=0pt, itemsep=0pt]
	\item either $R \cup C \connect{f + 1} L - F$,
	\item or $L \cup C \connect{f + 1} R - F$.
\end{enumerate}

\begin{lemma}
	If a graph $G$ is $2f$-connected, then it is $f$-good.
\end{lemma}
\begin{proof}
	Since $G$ is $2f$-connected so it is of size $n \ge 2f + 1$.
	We show the contrapositive that if a graph $G$ of size $n \ge 2f + 1$ is not $f$-good, then it has connectivity less than $2f$.
	Since $G$ is not $f$-good, there exists a set $F$ of cardinality at most $f$ and an $F$-partition $\del{ L, C, R }$ of $V$ such that
	\begin{enumerate}[topsep=0pt, itemsep=0pt]
		\item $R \cup C \notconnect{f + 1} L - F$, and
		\item $L \cup C \notconnect{f + 1} R - F$.
	\end{enumerate}
	There are three cases to consider:
	\begin{itemize}
		\item[Case 1:] $\envert{R \cup C} \le f$ and $\envert{L \cup C} \le f$.\\
		We have that $2f + 1 \le n = \envert{L \cup C \cup R} \le \envert{L \cup C} + \envert{R \cup C} \le 2f$, a contradiction.

		\item[Case 2:] $\envert{R \cup C} > f$ and $\envert{L \cup C} \le f$.
		Since $R \cup C \notconnect{f + 1} L - F$, we have that $\del{R \cup C} - \Gamma(L - F)$ is non-empty.
		Also, since $L - F$ is non-empty and $\envert{L} \le \envert{L \cup C} \le f$, we have that $\envert{L \cap F} < f$.
		Now $L - F$ has neighbors (outside of $L - F$) in either $L \cap F$ or $R \cup C$.
		Furthermore, we have that both $L - F$ and $\del{ R \cup C } - \Gamma( L - F )$ are non-empty.
		Therefore, removing $\Gamma(L - F) \cap \del{ R \cup C }$ and $L \cap F$ disconnects $L - F$ from $\del{ R \cup C } - \Gamma( L - F )$.
		To see that less than $2f$ nodes have been removed, observe that $\envert{ \Gamma(L - F) \cap \del{ R \cup C } } \le f$ (since $R \cup C \notconnect{f + 1} L - F$) and $\envert{ L \cap F } < f$.

		\item[Case 3:] $\envert{R \cup C} > f$ and $\envert{L \cup C} > f$.
		Since $R \cup C \notconnect{f + 1} L - F$ and $L \cup C \notconnect{f + 1} R - F$, we have that both $\del{R \cup C} - \Gamma(L - F)$ and $\del{ L \cup C } - \Gamma( R - F )$ are non-empty.
		WLOG we assume that $\envert{ L \cap F } \le \envert{ R \cap F }$ so that $\envert{ L \cap F } \le \floor{ \frac{f}{2} }$.
		Now $L - F$ has neighbors (outside of $L - F$) in either $L \cap F$ or $R \cup C$.
		Furthermore, we have that both $L - F$ and $\del{ R \cup C } - \Gamma( L - F )$ are non-empty.
		Therefore, removing $\Gamma(L - F) \cap \del{ R \cup C }$ and $L \cap F$ disconnects $L - F$ from $\del{ R \cup C } - \Gamma( L - F )$.
		To see that at most $\floor{ \frac{3f}{2} } \le 2f$ nodes have been removed, observe that $\envert{ \Gamma(L - F) \cap \del{ R \cup C } } \le f$ (since $R \cup C \notconnect{f + 1} L - F$) and $\envert{ L \cap F } \le \floor{ \frac{f}{2} }$ for a total of at most $f + \floor{\frac{f}{2}} = \floor{\frac{3f}{2}}$ nodes.
	\end{itemize}
\end{proof}

\begin{lemma} \label{lemma f good implies connectivity}
	If a graph $G$ is $f$-good, then it is $\del[0]{ \floor{ \frac{3f}{2} } + 1 }$-connected.
\end{lemma}
\begin{proof}
	We show that contrapositive that if a graph $G$ is not $\del[0]{ \floor{ \frac{3f}{2} } + 1 }$-connected, then it is not $f$-good.
	Since $G$ has connectivity less than $\del[0]{ \floor{ \frac{3f}{2} } + 1 }$, there exists a vertex cut $T$ of size at most $\floor{ \frac{3f}{2} }$ that separates two sets $A, B \subset V$.
	We create a set $F$ and a corresponding $F$-partition that violates the conditions for $f$-good to show that $G$ is not $f$-good.
	Partition the cut into $3$ roughly equal parts $T_1, T_2, T_3$ so that $\envert{T_1} \le \ceil{ \frac{f}{2} }$, $\envert{T_2} \le \floor{ \frac{f}{2} }$, and $\envert{T_3} \le \floor{ \frac{f}{2} }$.
	Let $F = T_1 \cup T_2$ so that $\envert{F} \le f$.
	We create an $F$-partition $\del{ L, C, R }$ as follows.
	Let $L = A \cup T_1$, $R = B \cup T_2$, and $C = T_3$.
	Observe that $\del{ L, C, R }$ is indeed a partition of $V$ and $L - F = A$ and $R - F = B$ are non-empty.
	Now we have that $\Gamma( L - F ) = \Gamma(A) = T$ and $\Gamma( R - F ) = \Gamma(B) = T$.
	Therefore $\del{ R \cup C } - \Gamma( L - F ) = \del{ B \cup T_2 \cup T_3 } - T = B$ is non-empty and $\del{ L \cup C } - \Gamma( R - F ) = \del{ A \cup T_1 \cup T_3 } - T = A$ is non-empty.
	Thus $\del{L, C, R}$ is indeed an $F$-partition.
	Now observe that $\Gamma( L - F ) \cap \del{R \cup C} = T \cap \del{ T_2 \cup T_3 \cup B } = T_2 \cup T_3$ has cardinality at most $\floor{ \frac{f}{2} } + \floor{ \frac{f}{2} } \le f$.
	Similarly $\Gamma( R - F ) \cap \del{L \cup C} = T \cap \del{ T_1 \cup T_3 \cup A } = T_1 \cup T_3$ has cardinality at most $\ceil{ \frac{f}{2} } + \floor{ \frac{f}{2} } = f$.
	Therefore we conclude that
	\begin{enumerate}[topsep=0pt, itemsep=0pt]
		\item $R \cup C \notconnect{f + 1} L - F$, and
		\item $L \cup C \notconnect{f + 1} R - F$.
	\end{enumerate}
	This completes the proof that $G$ is not $f$-good.
\end{proof}

\begin{lemma} \label{lemma f good implies degree}
	If a graph $G$ is $f$-good, then each node in $G$ has degree at least $2f$.
\end{lemma}
\begin{proof}
	First note that if $G$ is of size $n \le 2f$, then $G$ cannot be $f$-good (partition the graph into two parts $L$ and $R$ of roughly equal size $\le f$ with $F, C = \emptyset$).
	Therefore we assume that $n \ge 2f + 1$.
	We show the contrapositive that if there exists a node in $G$ of degree less than $2f$, then $G$ is not $f$-good.
	Suppose $u$ has degree strictly less than $2f$.
	We create a set $F$ and a corresponding $F$-partition that violates the conditions for $f$-good to show that $G$ is not $f$-good.
	Let $F$ be arbitrary $f-1$ neighbors of $u$ or the set $\Gamma(u)$ if $u$ has less than $f-1$ neighbors.
	Let $L = \set{u} \cup F$, $R = V - F - \set{u}$, and $C = \emptyset$.
	Then $\del{L, R, C}$ is indeed a partition of $V$, $L - F = \set{u}$ is non-empty, and $R - F = R$ is non-empty since it has at least $f + 1$ nodes (since $n \ge 2f + 1$, $\envert{L} = f$, and $\envert{C} = 0$).
	We show that
	\begin{enumerate}[topsep=0pt, itemsep=0pt]
		\item $R \cup C \notconnect{f + 1} L - F$, and
		\item $L \cup C \notconnect{f + 1} R - F$.
	\end{enumerate}
	Now $u = L - F$ has at most $f$ neighbors in $R \cup C$ by construction.
	Also $\envert{L \cup C} = \envert{L} \le f$ so that $R - F = R$ can not have more than $f$ neighbors in $L \cup C$.
\end{proof}

\begin{lemma} \label{lemma connectivity and degree imply f good}
	If a graph $G$ is $\del[0]{ \floor{ \frac{3f}{2} } + 1 }$-connected and every node has degree at least $2f$, then it is $f$-good.
\end{lemma}
\begin{proof}
	We show the contrapositive that if a graph $G$ is not $f$-good, then either it has connectivity at most $\floor{ \frac{3f}{2} }$ or there exists a node in $G$ of degree strictly less than $2f$.
	Since $G$ is not $f$-good, there exists a set $F$ of cardinality at most $f$ and an $F$-partition $\del{ L, C, R }$ of $V$ such that
	\begin{enumerate}[topsep=0pt, itemsep=0pt]
		\item $R \cup C \notconnect{f + 1} L - F$, and
		\item $L \cup C \notconnect{f + 1} R - F$.
	\end{enumerate}
	There are three cases to consider:
	\begin{itemize}[topsep=0pt, itemsep=0pt]
		\item[Case 1:] $\envert{R \cup C} \le f$ and $\envert{L \cup C} \le f$.\\
		We have that $n = \envert{L \cup C \cup R} \le \envert{L \cup C} + \envert{R \cup C} \le 2f$.
		Therefore, all nodes in $G$ have degree at most $2f - 1$.

		\item[Case 2:] $\envert{R \cup C} > f$ and $\envert{L \cup C} \le f$.
		Consider a node $u \in L - F$ (recall that $L - F$ is non-empty).
		We show that $u$ has degree strictly less than $2f$.
		Now $u$ has neighbors either in $L$ or in $R \cup C$.
		Observe that $\envert{L - u} = \envert{L} - 1 < \envert{L \cup C} \le f$ so that $u$ has strictly less than $f$ neighbors in $L$.
		Also $\envert{ \Gamma( L - F ) \cap \del{ R \cup C } } \le f$ since $R \cup C \notconnect{f + 1} L - F$.
		Therefore $u$ has at most $f$ neighbors in $R \cup C$.
		Thus $\envert{ \Gamma(u) } < 2f$, as required.

		\item[Case 3:] $\envert{R \cup C} > f$ and $\envert{L \cup C} > f$.
		Since $R \cup C \notconnect{f + 1} L - F$ and $L \cup C \notconnect{f + 1} R - F$, we have that both $\del{R \cup C} - \Gamma(L - F)$ and $\del{ L \cup C } - \Gamma( R - F )$ are non-empty.
		WLOG we assume that $\envert{ L \cap F } \le \envert{ R \cap F }$ so that $\envert{ L \cap F } \le \floor{ \frac{f}{2} }$.
		Now $L - F$ has either neighbors in $L \cap F$ or in $R \cup C$.
		Furthermore, we have that both $L - F$ and $\del{ R \cup C } - \Gamma( L - F )$ are non-empty.
		Therefore, removing $\Gamma(L - F) \cap \del{ R \cup C }$ and $L \cap F$ disconnects $L - F$ from $\del{ R \cup C } - \Gamma( L - F )$.
		To see that at most $\floor{ \frac{3f}{2} }$ nodes have been removed, observe that $\envert{ \Gamma(L - F) \cap \del{ R \cup C } } \le f$ (since $R \cup C \notconnect{f + 1} L - F$) and $\envert{ L \cap F } \le \floor{ \frac{f}{2} }$ for a total of at most $f + \floor{\frac{f}{2}} = \floor{\frac{3f}{2}}$ nodes.
	\end{itemize}
	This completes the proof that $G$ has either connectivity less than $\floor{ \frac{3f}{2} }$ or a node of degree less than $2f$.
\end{proof}

\begin{theorem}
	A graph $G$ is $f$-good if and only if it is $\del[0]{ \floor{ \frac{3f}{2} } + 1 }$-connected and every node has degree at least $2f$.
\end{theorem}
\begin{proof}
	From Lemmas \ref{lemma f good implies connectivity}, \ref{lemma f good implies degree}, and \ref{lemma connectivity and degree imply f good}.
\end{proof}
\end{document}